\documentclass[a4paper,UKenglish,cleveref, autoref, thm-restate,authorcolumns]{lipics-v2019}

\pdfoutput=1
\hideLIPIcs

\usepackage{graphicx}
\usepackage{booktabs}
\usepackage{listings}
\usepackage[ruled,vlined]{algorithm2e}
\usepackage{algpseudocode}
\usepackage{amsmath}
\usepackage{amsthm}
\usepackage{amssymb}
\usepackage{cleveref}
\usepackage{tikz}
\usepackage{placeins}

\usetikzlibrary{decorations, decorations.pathreplacing, decorations.pathmorphing,
    calc, backgrounds, positioning, tikzmark, patterns, matrix, fit}

\definecolor{chromaygb3}{rgb}{0.407843,0.776471,0.560784}
\definecolor{codeshadingcolour}{rgb}{.9,.9,.9}

\algnewcommand{\LeftComment}[1]{\(\triangleright\) #1}
\newcommand{\AlgVar}[1]{\mathit{#1}}

\SetKw{Continue}{continue}

\newcommand\Edge{\text{--}}

\newcommand\TDGeneral{\FuncSty{elimination\_forest}}
\newcommand\TDConnected{\FuncSty{elimination\_tree}}
\newcommand\TDOptimise{\FuncSty{optimise}}

\newcommand{\codelineref}[1]{line~\ref{#1}}
\newcommand{\linerangeref}[2]{lines~\ref{#1} to~\ref{#2}}

\title{An Algorithm for the Exact Treedepth Problem}

\author{James Trimble}{School of Computing Science, University of Glasgow \\ Glasgow, Scotland, UK }{j.trimble.1@research.gla.ac.uk}{https://orcid.org/0000-0001-7282-8745}{This work was supported by the Engineering and Physical Sciences Research Council (grant number EP/R513222/1).}%TODO mandatory, please use full name; only 1 author per \author macro; first two parameters are mandatory, other parameters can be empty. Please provide at least the name of the affiliation and the country. The full address is optional

\authorrunning{J. Trimble} %TODO mandatory. First: Use abbreviated first/middle names. Second (only in severe cases): Use first author plus 'et al.'

\Copyright{James Trimble} %TODO mandatory, please use full first names. LIPIcs license is "CC-BY";  http://creativecommons.org/licenses/by/3.0/


 %TODO mandatory; please add comma-separated list of keywords

\relatedversion{} %optional, e.g. full version hosted on arXiv, HAL, or other respository/website
\supplement{Source code: \url{https://github.com/jamestrimble/treedepth-solver}}%optional, e.g. related research data, source code, ... hosted on a repository like zenodo, figshare, GitHub, ...

\acknowledgements{Thanks to Ciaran McCreesh, David Manlove, Patrick Prosser and the anonymous referees for their helpful feedback, and to Robert Ganian, Neha Lodha, and Vaidyanathan Peruvemba Ramaswamy for providing software for the SAT encoding.}%optional

\nolinenumbers %uncomment to disable line numbering

\EventEditors{Simone Faro and Domenico Cantone}
\EventNoEds{2}
\EventLongTitle{18th International Symposium on Experimental Algorithms (SEA 2020)}
\EventShortTitle{SEA 2020}
\EventAcronym{SEA}
\EventYear{2020}
\EventDate{June 16--18, 2020}
\EventLocation{Catania, Italy}
\EventLogo{}
\SeriesVolume{160}
\ArticleNo{16}
\begin{document}

\maketitle

\begin{abstract}
  We present a novel algorithm for the minimum-depth elimination tree problem, which is equivalent to the optimal treedepth decomposition problem.
  Our algorithm makes use of two cheaply-computed lower bound functions to prune the search tree, along with symmetry-breaking and domination rules. We present an empirical study showing that the algorithm outperforms the current state-of-the-art solver (which is based on a SAT encoding) by orders of magnitude on a range of graph classes.
\end{abstract}

\section{Introduction}

This paper presents a practical algorithm for finding an optimal treedepth decomposition of a graph.
A \emph{treedepth decomposition} of graph $G=(V,E)$ is a rooted forest $F$ with node set $V$, such that for
each edge $\{u,v\} \in E$, we have either that $u$ is an ancestor of $v$ or $v$ is an ancestor of $u$
in $F$.  The \emph{treedepth} of $G$ is the minimum depth of a treedepth decomposition of $G$, where
depth is defined as the maximum number of vertices along a path from the root of the tree to a leaf.

Treedepth is closely related to a number of other problems.  The treedepth of a
connected graph $G$ equals the minimum height of an elimination tree for $G$
(\cite{DBLP:books/daglib/0030491}, chapter 6), which equals
the graph's vertex ranking number \cite{DBLP:conf/stacs/DeogunKKM94}.
%(?? what about disconnected graphs?)
The treedepth of a graph $G$ is also equal to the minimum number of
colours in a centred colouring of $G$ \cite{DBLP:books/daglib/0030491}.

Finding an elimination tree of small height is applicable to the parallel Cholesky factorisation
of sparse matrices \cite{zmijewski1986parallel}.
Treedepth also has relevance to the design of fixed-parameter tractable (FPT) algorithms.  For example,
the Mixed Chinese Postman Problem is FPT when parameterised by treedepth, but W[1]-hard
when parameterised by treewidth or pathwidth \cite{DBLP:journals/siamdm/GutinJW16}.

The decision variant of the treedepth problem is NP-complete \cite{pothen1988complexity}.  However, it can
be solved in linear time if the input graph is a tree
\cite{DBLP:journals/ipl/Schaffer89}, and in polynomial time for interval graphs
\cite{Aspvall1994}, trapezoid graphs, permutation graphs and circular arc
graphs \cite{DBLP:journals/dam/DeogunKKM99}.  There is a polynomial-time approximation algorithm
for the problem that gives a result within $O(\log^2 n)$ of the optimal value.
The problem is fixed parameter
tractable with respect to both treewidth and treedepth \cite{DBLP:journals/siamdm/BodlaenderDJKKMT98,DBLP:conf/icalp/ReidlRVS14}.
%There is an $O(c^n)$ algorithm for treedepth for a value of $c$ smaller than 2.

Although numerous heuristics for finding good elimination trees have been designed
and implemented \cite{groer2012inddgo},
we are aware of only two existing implementations of \emph{exact} algorithms for minimum treedepth
decomposition; both of these are introduced in \cite{DBLP:conf/alenex/GanianLOS19,DBLP:journals/corr/abs-1911-12995}.  In
each case, the optimisation problem is solved as a sequence of decision problems,
with each decision problem encoded as an instance of the boolean satisfiability problem
and solved using a general-purpose SAT solver.

\subparagraph*{This paper's contribution.} This paper introduces a new algorithm for computing an optimal treedepth
decomposition.  The algorithm is self-contained and does not require an external solver,
although it can optionally use the graph-automorphism library Nauty to break symmetries.
The basic structure of the
algorithm is very simple.  To improve performance, three symmetry breaking and domination features
and two lower-bounding functions are added to the algorithm.  In a set of experiments, we show that our
algorithm is typically orders of magnitude faster than the current (SAT-based) state of the art.

\subparagraph*{Structure of the paper.} \Cref{sec:preliminaries} introduces concepts and notation.
\Cref{sec:algorithm} presents the core parts of our algorithm.  \Cref{sec:extrafeatures}
describes enhancements to the basic algorithm.  \Cref{sec:implementation} provides details of our implementation.
\Cref{sec:experiments} presents an experimental comparison with the existing SAT encoding for treedepth.
\Cref{sec:conclusion} concludes.

\section{Preliminaries}\label{sec:preliminaries}

Let $G=(V,E)$ be a graph, where we assume that the elements of $V$ are integers.
We write $V(G)$ and $E(G)$ to denote the vertex and edge sets of $G$.
The neighbourhood $N_G(v)$ of a vertex $v$ is the set of vertices that are adjacent to $v$.
For $S \subseteq V$, we denote by $G[S]$ the subgraph of $G$
induced by $S$; that is, $(S, \{\{u,v\} \in E \mid u,v \in S\})$.
We use the notation $G - v$ for the removal of one vertex and its incident edges;
that is, $G - v = G[V(G) \setminus \{v\}]$.

The concepts \emph{elimination forest} and \emph{elimination tree}
are defined recursively in terms of one another.  An elimination forest of graph $G$ is
a rooted forest with vertex set $V(G)$ composed of elimination trees of each of $G$'s connected components.
An elimination tree of a non-empty connected graph $C$ is a rooted tree $T$ with vertex set
$V(C)$.  If $C$ has only one vertex $v$, then $T$ is a tree containing only $v$.  Otherwise,
$T$ is formed by choosing a vertex $v \in V(C)$ as the root, finding an elimination forest
of $C - v$, and making the trees in that forest the child subtrees of $v$.

We will use the graph $G$ in Figure~\ref{fig:graph} to provide an example of an elimination
tree, and as our running example throughout the paper.
The second part of the figure shows an optimal elimination tree of $G$,
which has depth 4.  Observe that removing the root vertex of the tree ($5$) from $G$
splits the graph into two connected components, and each of these components corresponds to
one of the child subtrees of $5$ in the elimination tree.

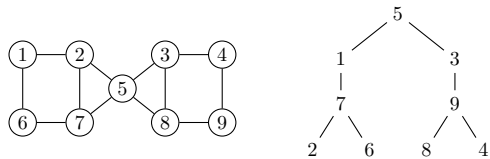
\begin{figure}[htb]
  \centering
\begin{tikzpicture}[scale=0.75, every node/.style={scale=0.75}]
    \node[shape=circle,draw=black,inner sep=2pt] (1) at (0,1.2) {1};
    \node[shape=circle,draw=black,inner sep=2pt] (2) at (1,1.2) {2};
    \node[shape=circle,draw=black,inner sep=2pt] (3) at (2.5,1.2) {3};
    \node[shape=circle,draw=black,inner sep=2pt] (4) at (3.5,1.2) {4};
    \node[shape=circle,draw=black,inner sep=2pt] (5) at (1.75,.6) {5};
    \node[shape=circle,draw=black,inner sep=2pt] (6) at (0,0) {6};
    \node[shape=circle,draw=black,inner sep=2pt] (7) at (1,0) {7};
    \node[shape=circle,draw=black,inner sep=2pt] (8) at (2.5,0) {8};
    \node[shape=circle,draw=black,inner sep=2pt] (9) at (3.5,0) {9};
    \node[] () at (0,-.6) {};   % just to move the picture up a bit

    \draw (1) -- (2);
    \draw (3) -- (4);
    \draw (6) -- (7);
    \draw (8) -- (9);

    \draw (1) -- (6);
    \draw (2) -- (7);
    \draw (3) -- (8);
    \draw (4) -- (9);

    \draw (2) -- (5);
    \draw (3) -- (5);
    \draw (7) -- (5);
    \draw (8) -- (5);
\end{tikzpicture}
\qquad
\begin{tikzpicture}[scale=0.75, every node/.style={scale=0.75}]
    \node[shape=circle,inner sep=2pt] (2) at (0,0) {2};
    \node[shape=circle,inner sep=2pt] (6) at (1,0) {6};
    \node[shape=circle,inner sep=2pt] (8) at (2,0) {8};
    \node[shape=circle,inner sep=2pt] (4) at (3,0) {4};
    \node[shape=circle,inner sep=2pt] (7) at (.5,0.8) {7};
    \node[shape=circle,inner sep=2pt] (9) at (2.5,0.8) {9};
    \node[shape=circle,inner sep=2pt] (1) at (.5,1.6) {1};
    \node[shape=circle,inner sep=2pt] (3) at (2.5,1.6) {3};
    \node[shape=circle,inner sep=2pt] (5) at (1.5,2.4) {5};

    \draw (5) -- (1);
    \draw (5) -- (3);
    \draw (1) -- (7);
    \draw (3) -- (9);
    \draw (7) -- (2);
    \draw (7) -- (6);
    \draw (9) -- (8);
    \draw (9) -- (4);
\end{tikzpicture}

  \caption{An example graph $G$ (left) and an optimal elimination tree of $G$ (right)}
  \label{fig:graph}
\end{figure}

Every elimination forest is a treedepth decomposition, and for a given graph $G$ there is
at least one elimination forest whose depth equals the treedepth of $G$
\cite{DBLP:books/daglib/0030491,mujika2015about}.
Therefore, in order to find an optimal treedepth decomposition of $G$ it is sufficient to search
for a mimumum-depth elimination tree of $G$.  This is the approach taken in this paper.

\section{The Algorithm}\label{sec:algorithm}

Our algorithm is shown as pseudocode in \Cref{TheAlgorithm}.  The shaded parts, and the third
parameter of each of the first two functions, will be introduced in later sections and can be disregarded for now.

{
\begin{algorithm}[htb]
 \footnotesize
\DontPrintSemicolon
\newcommand\SimpleLowerBound{\FuncSty{simple\_lower\_bound}}
\newcommand\PathLowerBound{\FuncSty{can\_prune\_by\_path\_lower\_bound}}

    \begin{tikzpicture}[remember picture,overlay]
        \coordinate (bounds1c) at ($(pic cs:bounds1) + (0, 0.15)$);
        \coordinate (bounds2c) at ($(pic cs:bounds2) + (0, -.2em)$);
        \node [fill=codeshadingcolour, rounded corners=.5ex, fit=(bounds1c) (bounds2c)] { };

        \coordinate (parent1c) at ($(pic cs:parent1) + (0, 0.15)$);
        \coordinate (parent2c) at ($(pic cs:parent2) + (0, -1.1em)$);
        \node [fill=codeshadingcolour, rounded corners=.5ex, fit=(parent1c) (parent2c)] { };

        \coordinate (parent3c) at ($(pic cs:parent3) + (0, 0.13)$);
        \coordinate (parent4c) at ($(pic cs:parent4) + (0, 0)$);
        \node [fill=codeshadingcolour, rounded corners=.5ex, fit=(parent3c) (parent4c)] { };

        \coordinate (symdom1c) at ($(pic cs:symdom1) + (0, 0.15)$);
        \coordinate (symdom2c) at ($(pic cs:symdom2) + (0, 0.02)$);
        \node [fill=codeshadingcolour, rounded corners=.5ex, fit=(symdom1c) (symdom2c)] { };
    \end{tikzpicture}

\nl $\TDGeneral(G,k,w)$ \label{td_general_fun} \;
\nl \KwData{Graph $G$, maximum depth $k$, and parent vertex $w$}
\nl \KwResult{$\AlgVar{true}$ if and only if an elimination forest of $G$ with depth $\leq k$ exists}
\nl \Begin{
  \nl \lIf{$k=0$ and $|V(G)|>0$}{\KwSty{return} $\AlgVar{false}$ \label{TDGeneralK0}}
  \nl $\AlgVar{\mathcal{C}} \gets$ the connected components of $G$ \label{MakeConnectedComponents} \;
  \nl   \tikzmark{bounds1}\For{$C \in \mathcal{C}$ \label{StartLowerBounds}}{
    \nl \lIf{$\SimpleLowerBound(|V(C)|) > k$}{\KwSty{return} $\AlgVar{false}$}
  }
  \nl   \For{$C \in \mathcal{C}$}{
    \nl \lIf{$\PathLowerBound(C, k)$}{\KwSty{return} $\AlgVar{false}$\tikzmark{bounds2} \label{EndLowerBounds}}
  }
  \nl   \For{$C \in \mathcal{C}$}{
    \nl \lIf{\KwSty{not} $\TDConnected(C, k, w)$\label{LoopToCallTDConnected}}{\KwSty{return} $\AlgVar{false}$}
  }
  \nl \KwSty{return} $\AlgVar{true}$ \;
}

\vspace{.5em}

\nl $\TDConnected(G,k,w)$ \label{td_connected_fun} \;
\nl \KwData{Connected, nonempty graph $G$, maximum depth $k\geq 1$, and parent vertex $w$}
\nl \KwResult{$\AlgVar{true}$ if and only if an elimination tree of $G$ with depth $\leq k$ exists}
\nl \Begin{
  \nl \If{$|V(G)|=1$ \label{SingleVertexGraph}}{
    \nl \tikzmark{parent1}$v \gets$ the unique element of $V(G)$\tikzmark{parent2} \;
    \nl $\AlgVar{parent}[v] \gets w$ \label{SetParentBaseCase} \;
    \nl \KwSty{return} $\AlgVar{true}$ \label{EndSingleVertexGraph} \;
  }
  \nl   \For{$v \in V(G)$ \label{VLoop}}{
    \nl \tikzmark{symdom1}\lIf{$v$ is ruled out by a symmetry or domination rule}{\Continue \tikzmark{symdom2}}
    \nl \tikzmark{parent3}$\AlgVar{parent}[v] \gets w$\tikzmark{parent4} \label{SetParent} \;
    \nl \lIf{$\TDGeneral(G - v, k - 1, v)$\label{CallTDGeneral}}{\KwSty{return} $\AlgVar{true}$}
  }
  \nl \KwSty{return} $\AlgVar{false}$ \;
}

\vspace{.5em}

\nl $\TDOptimise(G)$ \label{td_optimise_fun} \;
\nl \KwData{A graph $G$}
\nl \KwResult{The treedepth of $G$}
\nl \Begin{
  \nl $k \gets 0$ \;
  \nl \lWhile{$\TDGeneral(G,k,0) = \AlgVar{false}$}{$k \gets k + 1$}
  \nl \KwSty{return} $k$ \;
}
\caption{An algorithm to find an elimination forest of minimum depth.  To read the basic algorithm
(without optimisations), disregard the shaded sections and the third parameter
of each of the first two functions.}
\label{TheAlgorithm}
\end{algorithm}
}

\subparagraph*{Inputs and outputs.}
The algorithm's first function, $\TDGeneral()$, takes a graph $G$ and an integer $k \geq 0$, and returns
\emph{true} if and only if there exists an elimination forest of $G$ of depth $k$ or less.  The function
$\TDConnected()$ takes a connected, non-empty graph $G$ and an integer $k > 0$ and returns \emph{true}
if and only if there exists an elimination tree of depth $k$ or less.
The algorithm is run by calling $\TDOptimise()$, which takes a graph $G$ and returns the treedepth of $G$.

\subparagraph*{Details of the functions.}
The first two functions are mutually recursive, and closely follow
the definitions of elimination tree and forest.
The function $\TDGeneral()$ begins by returning \emph{false}---indicating
infeasibility---if a elimination tree of depth zero is sought for a non-empty graph.
The function then returns \emph{true} if and only if an elimination tree of depth no greater
than $k$ exists for each connected component of the graph.

The function $\TDConnected(G, k)$ returns \emph{true} if $G$ has a single vertex
(\linerangeref{SingleVertexGraph}{EndSingleVertexGraph}).  Otherwise,
it tries each vertex $v \in V(G)$ in turn (\codelineref{VLoop}), and returns \emph{true} if
and only if one of these $v$ values can be the root of an elimination tree of depth
$k$---which is the case if and only if 
an elimination forest of depth no greater than $k-1$ exists for $G-v$.

The main function, $\TDOptimise()$, carries  
out repeated calls to $\TDGeneral()$ with ascending values of $k$
until a feasible depth is reached.  It would be possible to implement a
branch-and-bound variant of the algorithm with a little additional effort, and it is likely that
this would be somewhat faster than the approach we have taken.  We decided not to do so for two
reasons.  First, having a sequence of decision problems simplifies the exposition of the algorithm.
Second, we have observed in practice that two of the decision problems---the final unsatisfiable problem
and the satisfiable problem after which the algorithm terminates---take up most of the run time.
This suggests that a branch-and-bound approach would be of limited benefit.

An inductive proof of the algorithm's correctness appears in \Cref{appendix:proof}.

\subparagraph*{Example.}
Returning to our example graph $G$ from Figure~\ref{fig:graph}, suppose $G$ is 
passed to $\TDConnected()$.  Figure~\ref{fig:twosubproblems} illustrates two of the nine subproblems explored at
\codelineref{VLoop}.  In the left part of the figure, $v=1$.  Choosing this vertex
as the root of the elimination tree leaves a connected graph, which is
passed to $\TDGeneral()$ at \codelineref{CallTDGeneral}.  The right part of the figure
corresponds to the decision $v=5$.  Choosing
this vertex as the root leaves a disconnected graph, with two four-vertex components.  This disconnected
graph is passed to $\TDGeneral()$ at \codelineref{CallTDGeneral}, and subsequently $\TDConnected()$ is
called for each of the two components.

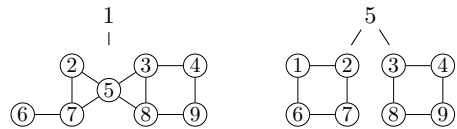
\begin{figure}[htb]
  \centering
\begin{tikzpicture}[scale=0.9, every node/.style={scale=0.9}]
    \node[] (graph1) at (.5,-.3) {
      \begin{tikzpicture}[scale=0.9, every node/.style={scale=0.9}]
          \node[shape=circle,draw=black,inner sep=1pt] (2) at (.8,.8) {2};
          \node[shape=circle,draw=black,inner sep=1pt] (3) at (2,.8) {3};
          \node[shape=circle,draw=black,inner sep=1pt] (4) at (2.8,.8) {4};
          \node[shape=circle,draw=black,inner sep=1pt] (5) at (1.4,.4) {5};
          \node[shape=circle,draw=black,inner sep=1pt] (6) at (0,0) {6};
          \node[shape=circle,draw=black,inner sep=1pt] (7) at (.8,0) {7};
          \node[shape=circle,draw=black,inner sep=1pt] (8) at (2,0) {8};
          \node[shape=circle,draw=black,inner sep=1pt] (9) at (2.8,0) {9};

          \draw (3) -- (4);
          \draw (6) -- (7);
          \draw (8) -- (9);

          \draw (2) -- (7);
          \draw (3) -- (8);
          \draw (4) -- (9);

          \draw (2) -- (5);
          \draw (3) -- (5);
          \draw (7) -- (5);
          \draw (8) -- (5);
      \end{tikzpicture}
    };
    \node[shape=circle,inner sep=2pt] (1) at (.5,0.8) {1};
    \draw (1) -- (graph1);
\end{tikzpicture}
\qquad
\begin{tikzpicture}[scale=0.9, every node/.style={scale=0.9}]
    \node[] (graph1) at (-.2,-.3) {
      \begin{tikzpicture}[scale=0.9, every node/.style={scale=0.9}]
          \node[shape=circle,draw=black,inner sep=1pt] (1) at (0,.8) {1};
          \node[shape=circle,draw=black,inner sep=1pt] (2) at (.8,.8) {2};
          \node[shape=circle,draw=black,inner sep=1pt] (6) at (0,0) {6};
          \node[shape=circle,draw=black,inner sep=1pt] (7) at (.8,0) {7};
          \draw (1) -- (2);
          \draw (1) -- (6);
          \draw (2) -- (7);
          \draw (6) -- (7);
      \end{tikzpicture}
    };
    \node[] (graph2) at (1.2,-.3) {
      \begin{tikzpicture}[scale=0.9, every node/.style={scale=0.9}]
          \node[shape=circle,draw=black,inner sep=1pt] (3) at (0,.8) {3};
          \node[shape=circle,draw=black,inner sep=1pt] (4) at (.8,.8) {4};
          \node[shape=circle,draw=black,inner sep=1pt] (8) at (0,0) {8};
          \node[shape=circle,draw=black,inner sep=1pt] (9) at (.8,0) {9};
          \draw (3) -- (4);
          \draw (8) -- (9);
          \draw (3) -- (8);
          \draw (4) -- (9);
      \end{tikzpicture}
    };
    \node[shape=circle,inner sep=2pt] (5) at (.5,0.8) {5};
    \draw (5) -- (graph2);
    \draw (5) -- (graph1);
\end{tikzpicture}
  \caption{Two of the nine subproblems visited by the first call to $\TDConnected$
    on our example graph}
  \label{fig:twosubproblems}
\end{figure}

%%It can be shown by induction on $k$ that $\TDGeneral()$
%%terminates with the correct answer.  In the base case $k=0$, the function correctly returns
%%\emph{true} if and only if $G$ is empty.  Now, for the inductive case, let $k>0$ be given
%%and suppose that $\TDGeneral(G, k')$ is sound and terminates for all $k'<k$.
%%In the trivial case of an empty graph, $\TDGeneral(G, k')$ correctly returns
%%\emph{true}.  Otherwise, $\TDConnected(C, k)$ is called for each connected
%%component, which is valid by \ref{lemma1}.  $\TDConnected(C, k)$ in turn calls
%%$\TDGeneral()$ for each possible vertex removal, with $k-1$ as the depth argument.
%%The strategy of trying all vertex removals is valid by \ref{lemma2}, and the recursive
%%call to $\TDGeneral()$ is valid by induction.

\subsection{Generating an Elimination Forest}\label{sec:generating}

The algorithm described so far returns only a single integer: the treedepth of the input graph.
We can easily modify the algorithm to also produce an elimination forest of that depth.

We assume that vertices of the input graph $G$ are numbered from $1$ to $|V(G)|$. A global array
with $|V(G)|$ elements, $\AlgVar{parent}$, is used to record the parent of each vertex
in the elimination forest.  A value $\AlgVar{parent}[v] = 0$ indicates that vertex $v$ is a root.
Lines \ref{SetParentBaseCase} and \ref{SetParent} of \cref{TheAlgorithm} record the parent of $v$.

The functions $\TDGeneral()$ and $\TDConnected()$ both have $w$ as an extra parameter.
Vertex $w$ will be parent to the first vertex chosen from $G$.  At the first level
of recursion for either of these functions, $w=0$, indicating that the next vertex to be selected will be
a root.

When either $\TDGeneral()$ and $\TDConnected()$ returns $\AlgVar{true}$, this indicates
not only that an elimination tree of depth $k$ of the subgraph $G$ exists, but also that such
a decomposition has been recorded in the $\AlgVar{parent}$ array (with any $\AlgVar{parent}$ value
not in $V(G)$ indicating a root of the subproblem's decomposition).

\section{Enhancements to the Algorithm: Symmetry Breaking and Domination Rules, Pruning, and Sorting}\label{sec:extrafeatures}

In this section we describe five improvements that can be made to the basic algorithm.
The first three of these are symmetry breaking and domination rules that
allow us to avoid choosing some values of $v$ at \codelineref{VLoop} of
$\TDConnected()$.
The fourth technique allows us to prune subproblems by quickly computing a lower bound
on treedepth; we introduce two such bounds.  The final technique is a re-ordering of vertices before solving.

\subsection{Symmetry Breaking and Domination Rules}

Recall that the loop at \codelineref{VLoop} of $\TDConnected()$ tries each vertex $v$ as a
potential root of the elimination tree, attempting to find an elimination tree of depth $k$ or less.
For some graphs, there are several values of $v$ that may be chosen as the root of such a tree;
let $S$ be the set of such vertices.  We can omit some choices of
$v$ at line \codelineref{VLoop} without affecting the algorithm's correctness,
provided at least one member of $S$ is chosen.  The three
symmetry-breaking and domination rules that follow make use of this fact to avoid visiting
some vertices. They ensure that the least-numbered vertex in $S$ is visited, if $S$ is non-empty.

\subparagraph*{Symmetry breaking using vertex orbits}

Our first symmetry breaking technique is applied only to connected graphs.
Before beginning the algorithm, we use Nauty \cite{McKay201494} to compute the orbits
of the vertices.  (Recall that vertices $v$ and $v'$ are in the same orbit if and
only if there is an automorphism that maps $v$ to $v'$.)  In the first call to
$\TDConnected()$---that is, the call where $G$ is the full input graph---we
can avoid choosing any value of $v$ in the loop if there exists a vertex $v' < v$ that is in the
same orbit as $v$.  This symmetry-breaking technique is valid
because the subgraph created by the removal of $v$ is isomorphic
to the subgraph created by the removal of $v'$, and thus has the same treedepth.
%?? If the code is moved, modify the description slightly

As an example, consider the graph in Figure~\ref{fig:graph}.  Since vertices $1$, $4$, $6$, and $9$
are in the same orbit, we can avoid choosing vertices $4$, $6$, and $9$ in the
first call to $\TDConnected()$.

This technique is useful on highly symmetrical graphs, but of course cannot
be expected to achieve a speedup of more than $|V(G)|$.
A number of avenues for improved symmetry breaking could be explored in future work.
It would be straightforward
to extend the symmetry breaking to disconnected graphs by computing the orbits of
vertices in each connected component separately.  Our technique for symmetry breaking
could also be used for subproblems; this may be useful
for very symmetrical graphs, but we suspect that the cost of additional calls to Nauty would outweigh
the benefit in many cases.   We could move beyond finding the orbits of
single vertices; it is possible to find the orbits of pairs (or larger sets)
of vertices with a single call to Nauty and a small amount of extra work.
Lastly, we could use symmetry-breaking techniques from constraint programming
such as GE-trees \cite{DBLP:conf/ecai/Roney-DougalGKL04}.

\subparagraph*{Vertex domination}

Suppose we have $v, v' \in V(G)$ (not necessarily adjacent) such that $v' < v$ and
$N_G(v') \setminus \{v\} \supseteq N_G(v) \setminus \{v'\}$.  We say that
$v'$ \emph{dominates} $v$.  Clearly,
$G - v'$ is isomorphic to a subgraph of $G - v$.  Thus, the minimum depth of a treedepth
decomposition of $G$ rooted at $v$ is no smaller that the minimum depth of a
decomposition rooted at $v'$.  We can use this fact in $\TDConnected()$
(at any depth of recursion) by rejecting at \codelineref{VLoop} any value of $v$ such that there exists
a lower-numbered $v'$ such that $N_G(v') \setminus \{v\} \supseteq N_G(v) \setminus \{v'\}$.

As an example, suppose the input graph is a clique $K_n$, with the vertices numbered
$\{1, \dots, n\}$.  At each call to $\TDConnected()$, the lowest-numbered
vertex in $G$ dominates the other vertices; thus only one vertex needs to be explored
in the loop at \codelineref{VLoop}.

Our use of vertex domination is based on \cite{DBLP:conf/alenex/GanianLOS19,DBLP:journals/corr/abs-1911-12995}, where
the technique is used in a preprocessing step to generate constraints for the input graph,
rather than applied to each subproblem.

\subparagraph*{Only-child vertices}

Consider again our example graph in Figure~\ref{fig:graph}.  Suppose that the
symmetry-breaking and domination rules described so far in this section are disabled, and that
a decomposition of depth 3 is being sought.  
Figure~\ref{fig:onlychildexample}(a) shows the program state after selecting vertex
$6$ as the root vertex of the tree and vertex $5$ as its child.  There are two
subproblems: the subgraphs induced by $\{1,2,7\}$ and $\{3,4,8,9\}$.  We call $5$
an \emph{only child} in the tree because it has a parent (vertex $6$) but has no siblings.

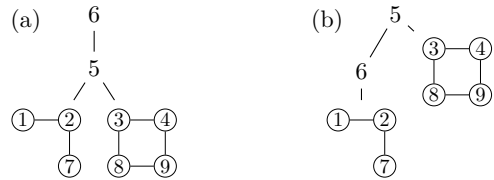
\begin{figure}[htb]
  \centering
\begin{tikzpicture}[scale=0.9, every node/.style={scale=0.9}]
    \node[] (graph1) at (-.2,-.3) {
      \begin{tikzpicture}[scale=0.85, every node/.style={scale=0.85}]
          \node[shape=circle,draw=black,inner sep=1pt] (1) at (0,.8) {1};
          \node[shape=circle,draw=black,inner sep=1pt] (2) at (.8,.8) {2};
          \node[shape=circle,draw=black,inner sep=1pt] (7) at (.8,0) {7};
          \draw (1) -- (2);
          \draw (2) -- (7);
      \end{tikzpicture}
    };
    \node[] (graph2) at (1.2,-.3) {
      \begin{tikzpicture}[scale=0.85, every node/.style={scale=0.85}]
          \node[shape=circle,draw=black,inner sep=1pt] (3) at (0,.8) {3};
          \node[shape=circle,draw=black,inner sep=1pt] (4) at (.8,.8) {4};
          \node[shape=circle,draw=black,inner sep=1pt] (8) at (0,0) {8};
          \node[shape=circle,draw=black,inner sep=1pt] (9) at (.8,0) {9};
          \draw (3) -- (4);
          \draw (8) -- (9);
          \draw (3) -- (8);
          \draw (4) -- (9);
      \end{tikzpicture}
    };
    \node[shape=circle,inner sep=2pt] (5) at (.5,0.8) {5};
    \node[shape=circle,inner sep=2pt] (6) at (.5,1.6) {6};
    \draw (6) -- (5);
    \draw (5) -- (graph2);
    \draw (5) -- (graph1);
    \node[shape=circle,inner sep=2pt] (a) at (-.5,1.5) {(a)};
\end{tikzpicture}
\qquad
\qquad
\begin{tikzpicture}[scale=0.9, every node/.style={scale=0.9}]
    \node[] (graph1) at (0,-.3) {
      \begin{tikzpicture}[scale=0.85, every node/.style={scale=0.85}]
          \node[shape=circle,draw=black,inner sep=1pt] (1) at (0,.8) {1};
          \node[shape=circle,draw=black,inner sep=1pt] (2) at (.8,.8) {2};
          \node[shape=circle,draw=black,inner sep=1pt] (7) at (.8,0) {7};
          \draw (1) -- (2);
          \draw (2) -- (7);
      \end{tikzpicture}
    };
    \node[] (graph2) at (1.4,.75) {
      \begin{tikzpicture}[scale=0.85, every node/.style={scale=0.85}]
          \node[shape=circle,draw=black,inner sep=1pt] (3) at (0,.8) {3};
          \node[shape=circle,draw=black,inner sep=1pt] (4) at (.8,.8) {4};
          \node[shape=circle,draw=black,inner sep=1pt] (8) at (0,0) {8};
          \node[shape=circle,draw=black,inner sep=1pt] (9) at (.8,0) {9};
          \draw (3) -- (4);
          \draw (8) -- (9);
          \draw (3) -- (8);
          \draw (4) -- (9);
      \end{tikzpicture}
    };
    \node[shape=circle,inner sep=2pt] (5) at (.5,1.6) {5};
    \node[shape=circle,inner sep=2pt] (6) at (0,0.75) {6};
    \draw (5) -- (6);
    \draw (5) -- (graph2);
    \draw (6) -- (graph1);
    \node[shape=circle,inner sep=2pt] (b) at (-.5,1.5) {(b)};
\end{tikzpicture}
  \caption{The only child rule allows us to avoid the effort of exploring the subproblem
  in the left part of the figure, since at least as good a decomposition can be achieved
  by choosing $5$ as the root vertex.}
  \label{fig:onlychildexample}
\end{figure}

Figure~\ref{fig:onlychildexample}(b) shows the tree if vertex $5$
is chosen before, rather than after, vertex $6$.  Observe that the same two subproblems
appear, but one of them has been lifted to a higher level in the tree.  It is clear
that the optimal elimination tree based on Figure~\ref{fig:onlychildexample}(a)
will have depth no less
than that of the optimal elimination tree based on Figure~\ref{fig:onlychildexample}(b).

This is an example of a rule that holds in general.

\begin{proposition}\label{lemma1}
  (Only-child rule) Let $T$ be a depth-$k$ treedepth decomposition of a connected graph $G$.
  If the root vertex of $T$ has an only child $v'$, then there exists a treedepth decomposition
  of $G$ of depth no greater than $k$ with $v'$ as its root.
\end{proposition}
\begin{proof}
Let $v$ be the root of $T$, and $v'$ its only child.
Let $\{G_1, \dots, G_a\}$ be the child subproblems that contain a vertex adjacent to $v$
in $G$, and let $\{H_1, \dots, H_b\}$ be the child subproblems that do not,
as illustrated in Figure~\ref{fig:onlychildgeneral} (either of these sets of subproblems
may be empty).
For each graph in $\{G_1, \dots, G_a\} \cup \{H_1, \dots, H_b\}$, there must exist
a decomposition of depth at most $k-2$.

If we reverse the order of $v$ and $v'$,
then the tree and its subproblems will be as shown in the second part of Figure~\ref{fig:onlychildgeneral}
(where $\{H_1, \dots, H_b\}$ are moved up a level).  Clearly, it is possible to construct a decomposition
with depth no greater than $k$ by using the same decompositions of the
subproblems $\{G_1, \dots, G_a\} \cup \{H_1, \dots, H_b\}$ as in $T$.
\end{proof}

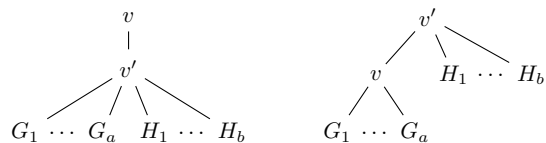
\begin{figure}[htb]
  \centering
\begin{tikzpicture}[scale=0.85, every node/.style={scale=0.85}]
    \node[] (G1) at (-1.6,-.2) {$G_1$};
    \node[] (Gdots) at (-1,-.2) {$\dots$};
    \node[] (Ga) at (-.4,-.2) {$G_a$};
    \node[] (H1) at (.4,-.2) {$H_1$};
    \node[] (Hdots) at (1,-.2) {$\dots$};
    \node[] (Hb) at (1.6,-.2) {$H_b$};
    \node[shape=circle,inner sep=2pt] (vprime) at (0,0.8) {$v'$};
    \node[shape=circle,inner sep=2pt] (v) at (0,1.6) {$v$};
    \draw (v) -- (vprime);
    \draw (vprime) -- (G1);
    \draw (vprime) -- (Ga);
    \draw (vprime) -- (H1);
    \draw (vprime) -- (Hb);
\end{tikzpicture}
\qquad
\begin{tikzpicture}[scale=0.85, every node/.style={scale=0.85}]
    \node[] (G1) at (-1.4,-.2) {$G_1$};
    \node[] (Gdots) at (-.8,-.2) {$\dots$};
    \node[] (Ga) at (-.2,-.2) {$G_a$};
    \node[] (H1) at (.4,.7) {$H_1$};
    \node[] (Hdots) at (1,.7) {$\dots$};
    \node[] (Hb) at (1.6,.7) {$H_b$};
    \node[shape=circle,inner sep=2pt] (v) at (-.8,.7) {$v$};
    \node[shape=circle,inner sep=2pt] (vprime) at (0,1.6) {$v'$};
    \draw (v) -- (vprime);
    \draw (v) -- (G1);
    \draw (v) -- (Ga);
    \draw (vprime) -- (H1);
    \draw (vprime) -- (Hb);
\end{tikzpicture}
  \caption{The general case of the only child rule.  If a vertex
  $v$ has an only child $v'$ in the elimination tree, then an elimination
  tree of the no greater depth can be found by reversing the positions of $v$
  and $v'$.}
  \label{fig:onlychildgeneral}
\end{figure}

This can be used for a simple domination breaking rule.  If the removal of a vertex $v$ chosen at
\codelineref{VLoop} of \cref{TheAlgorithm} leaves a non-empty, connected graph, then the child
vertex of $v$ in the treedepth decomposition, $v'$, must be an only child.  The only-child rule
allows us to omit any vertex $v'$ that has a lower number than $v$.

\subsection{Computing Lower Bounds}

At \linerangeref{StartLowerBounds}{EndLowerBounds} of \cref{TheAlgorithm},
lower bounds are computed with the goal of quickly proving that one of the
subproblems is infeasible.  For each connected component $C$, two functions are
called to find lower bounds on the treedepth of $C$.  If either of these bounds
is greater than $k$, the current subproblem is unsatisfiable and
$\AlgVar{false}$ is returned.  The first lower-bounding function uses a simple and very fast
algorithm to compute a bound based on the number of vertices in the subproblem
and an upper bound on the maximum degree.  The second function greedily
constructs a path in the graph, and uses as a lower bound a well-known formula for the treedepth of a
path graph.

\subparagraph*{Simple lower bound}

Let $b > 0$ be an upper bound on the maximum degree of a graph $G$.
If $G$ has no vertices, then its treedepth is zero.  Otherwise, if we remove a single vertex and
its incident edges from $G$, it is clear that the resulting graph can
have at most $b$ connected components and at least one of these components must have
$\lceil (|V(G)| - 1) / b \rceil$ or more vertices.
\Cref{SimpleLowerBoundAlgorithm} is a recursive algorithm that makes use of this
fact to give a lower bound on the treedepth of a graph.

{
\begin{algorithm}[htb]
 \footnotesize
\DontPrintSemicolon
\newcommand\SimpleLowerBound{\FuncSty{simple\_lower\_bound}}
\nl $\SimpleLowerBound(n)$ \label{simple_lower_bound_fun} \;
\nl \Begin{
  \nl \lIf{$n = 0$}{\KwSty{return} $0$}
  \nl \KwSty{return} $1 + \SimpleLowerBound(\lceil (n - 1) / b \rceil )$ \;
}
\caption{The simple lower bound function}
\label{SimpleLowerBoundAlgorithm}
\end{algorithm}
}

In our implementation, $b$ is a global variable equal to the maximum degree of the
input graph.  If $b = 0$, we do not use this lower bounding technique.

This bounding algorithm runs in time $O(\log n)$, where $n$ is the argument passed
to the function.
As an optimisation, our implementation pre-computes $\FuncSty{simple\_lower\_bound}(n)$
for $n \in \{0, \dots, |V(G)|\}$, and saves these values in an array before calling
$\FuncSty{td\_optimise}()$, thus allowing the bounding function to run in constant time.
However, we use a bitset popcount (number of set bits) operation to calculate the value of $n$, and therefore
the overall time complexity of calculating this bound is $O(n)$, where $n$ is the size
of the input graph.

Our example graph in Figure~\ref{fig:graph} has $9$ vertices and maximum degree $4$.
The bounding function therefore gives a lower bound of 3 on the graph's treedepth.

\subparagraph*{Path lower bound}

A graph containing a path of $k$ vertices has treedepth at least
$\lceil \log_2(k+1)\rceil$ \cite{DBLP:books/daglib/0030491}.  We can thus cheaply
compute a lower bound on the treedepth of a graph $G$ by greedily
finding a path in $G$, as shown in \Cref{PathLowerBoundAlgorithm}.
We choose the lowest-numbered vertex $v$ in $G$ as our starting point, and
attempt to grow the path from $v$ in two directions (\codelineref{TwoDirectionsOfPathSearch}).
In each of these two growing phases, the algorithm extends the path
by one vertex at a time until the most-recently-visited vertex has no
neighbours that are not on the path (\linerangeref{StartOfPathLoop}{EndOfPathLoop}).
If $\lceil \log_2(k+1)\rceil$, where $k$ is the length of the constructed path,
exceeds the target treedepth, the algorithm returns $\AlgVar{true}$.

{
\begin{algorithm}[htb]
 \footnotesize
\DontPrintSemicolon
\newcommand\PathLowerBound{\FuncSty{can\_prune\_by\_path\_lower\_bound}}
% https://tex.stackexchange.com/questions/434664/algorithm2e-do-n-times
\SetKwFor{RepTimes}{repeat}{times}{end}
\nl $\PathLowerBound(G, \AlgVar{targetDepth})$ \label{path_lower_bound_fun} \;
\nl \Begin{
  \nl $v \gets \min(V(G))$ \;
  \nl $P \gets \{v\}$ \Comment{$P$ is the set of vertices on the path} \;
  \nl \RepTimes{2\label{TwoDirectionsOfPathSearch}} {
    \nl $u \gets v$ \;
    \nl \While {$u$ has a neighbour that is not contained in $P$ \label{StartOfPathLoop}}{
      \nl $u \gets$ the least such neighbour \;
      \nl $P \gets P \cup \{u\}$ \label{EndOfPathLoop} \;
    }
  }
  \nl \KwSty{return} $\lceil \log_2(|P| + 1) \rceil > \AlgVar{targetDepth}$ \label{PathPrune} \;
}
\caption{The path lower bound function}
\label{PathLowerBoundAlgorithm}
\end{algorithm}
}

As an example, consider again the graph in Figure~\ref{fig:graph}.  The
algorithm $\FuncSty{path\_lower\_bound}()$ begins with $v=1$, then greedily
extends the path with vertices $2, 5, 3, 4, 9,$ and $8$.  As it is not
possible to extend the path further from vertex $8$, the algorithm returns
to vertex $1$ and prepends vertices $6$ and $7$ to the path.  The path found
is thus $7\Edge 6\Edge 1\Edge 2\Edge 5\Edge 3\Edge 4\Edge 9\Edge 8$, which contains
all nine of the graph's vertices and gives a lower bound of $\lceil \log_2(9+1)\rceil = 4$.
(In this case, this equals the treedepth of the graph, so our optimisation algorithm is
able to determine that the graph has treedepth greater than 3 without any recursive calls
on subgraphs).

Since we use bitset operations to iterate over the neighbours of $u$ at \codelineref{StartOfPathLoop}
of \cref{PathLowerBoundAlgorithm}, the algorithm runs in $O(|V(G)|^2)$ time.

Our implementation has an additional small improvement to this lower bounding function which,
for simplicity, is not shown in \cref{PathLowerBoundAlgorithm}.  If the found path
has three or more vertices, our implementation looks for a pair of vertices $(v,w)$ that
appear in the path, such that $v$ and $w$ are neighbours in $G$ but do not appear
side-by-side in the path.
The section of the path from $v$ to $w$ thus forms a cycle, and it is possible to use the
bound $1 + \lceil \log_2(k)\rceil$, where $k$ is the number of vertices in the cycle \cite{DBLP:books/daglib/0030491}.
The algorithm continues to visit all such pairs of vertices, stopping early if the
calculated bound is greater than $\AlgVar{targetDepth}$.  This cycle-finding extension
of the bounding algorithm runs in $O(|V(G)|^2)$ time, and thus does not increase the algorithm's
time complexity.

\subparagraph*{Comparison of the two bounds}

Neither the simple lower bound nor the path lower bound dominates the other.  We have already
seen that for our example graph, the path lower bound is greater than the simple lower bound.
\Cref{fig:boundsgraph} is a graph for which the reverse is true.  The graph has order 7 and
maximum degree 3; therefore the simple lower bound is 3.  The path lower-bounding function
finds the path $1\Edge 5\Edge 2$ which has length 3, giving a bound of 2.

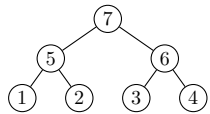
\begin{figure}[htb]
  \centering
\begin{tikzpicture}[scale=0.75, every node/.style={scale=0.75}]
    \node[shape=circle,draw=black,inner sep=2pt] (1) at (0,0) {1};
    \node[shape=circle,draw=black,inner sep=2pt] (2) at (1,0) {2};
    \node[shape=circle,draw=black,inner sep=2pt] (3) at (2,0) {3};
    \node[shape=circle,draw=black,inner sep=2pt] (4) at (3,0) {4};
    \node[shape=circle,draw=black,inner sep=2pt] (5) at (.5,0.7) {5};
    \node[shape=circle,draw=black,inner sep=2pt] (6) at (2.5,0.7) {6};
    \node[shape=circle,draw=black,inner sep=2pt] (7) at (1.5,1.4) {7};

    \draw (1) -- (5);
    \draw (2) -- (5);
    \draw (3) -- (6);
    \draw (4) -- (6);
    \draw (5) -- (7);
    \draw (6) -- (7);
\end{tikzpicture}

  \caption{A graph for which the simple lower bound is greater than the path lower bound}
  \label{fig:boundsgraph}
\end{figure}

\subsection{Initial Vertex Ordering by Degree}

Before running our algorithm, the vertices of the input graph are reordered by
non-increasing degree.  We have observed that this leads to speed-ups on graph
classes including binary trees and random graphs.  We give two speculative reasons for
this.  First, it seems likely that for satisfiable values of the treedepth
parameter $k$, we are most likely to find an elimination forest of depth $k$ quickly
by choosing high-degree vertices first, as these are most likely to split the remainder
of the graph into small components.  Second, the path-finding algorithm in our path lower
bound function chooses low-numbered vertices first, and by preferring high-degree vertices
it is less likely to run into ``dead ends''.

\section{Bitset Implementation}\label{sec:implementation}

We use bitsets to represent sets, including rows of adjacency matrices.  Induced subgraphs
are not stored explicitly in memory; our program simply passes a pointer to the full graph
along with a pointer to the set of vertices that induce the subgraph.

The space complexity of our algorithm compares favourably to that
of the partitioning-based SAT encoding, which uses $O(n^3)k$ clauses, where $n=|V(G)|$.
%?? But is this a tight upper bound?

\begin{proposition}\label{spacetheorem}
Using a bitset implementation, the algorithm requires $O(n^2)$ space.
\end{proposition}
\begin{proof}
  The bit-matrix representation of the input graph $G$ requires $O(n^2)$ space.

  At most
  $n+1$ calls are made to $\TDGeneral()$, and at most $n$ calls are made to
  $\TDConnected()$, since \codelineref{CallTDGeneral} of \cref{TheAlgorithm}
  passes a graph with one vertex fewer than the graph passed to $\TDConnected()$.
  Each recursive call to these functions requires $O(n)$ space, with the exception
  of the connected components found at \codelineref{MakeConnectedComponents}.  Since each
  connected component is stored in an $n$-element bitset, it remains to show that
  at most $O(n)$ connected components are stored at once.

  We prove by induction that no more than $n$ components are stored at one time.
  If $\TDGeneral()$ is called with a 1-vertex graph, then only one component
  is found, and no further components are found in recursive calls to the function.
  Now, let an $n$-vertex graph $G$ be given, and assume that $\TDGeneral()$ and its
  recursive calls create no more than $i$ components when called with any $i$-vertex
  graph ($1 \leq i \leq n$).  Let $c$ be the number of components created by
  the initial call $\TDGeneral(G, \dots)$.
  Each of these components has at most $n - c + 1$ vertices, since otherwise the
  components would have more than $n$ vertices in total.  If a component
  with $m$ vertices is passed to $\TDConnected()$, then the recursive
  call to $\TDGeneral()$ at \codelineref{CallTDGeneral} passes a graph
  with at most $m - 1 \leq n - c$ vertices.  By our inductive assumption, this
  call requires space for at most $n - c$ components, and therefore at most
  $c + n - c = n$ components in total are needed for the call to $\TDGeneral(G, \dots)$.
\end{proof}

\section{Experiments}\label{sec:experiments}

This section presents an experimental comparison with the partition-based
SAT encoding \cite{DBLP:conf/alenex/GanianLOS19,DBLP:journals/corr/abs-1911-12995} which is the existing state
of the art for the exact treedepth
problem.  We also investigate the effect
of switching off individual features of our algorithm.

The experiments were performed
on a cluster of five machines with dual Intel Xeon E5-2697A v4 CPUs and 512 GBytes
of RAM, running Ubuntu 18.04.  We implemented our algorithm in C, using Nauty version 2.6r11 for vertex-orbit symmetry breaking.
The program for the SAT encoding\footnote{\url{https://github.com/nehal73/TCW_TD_to_SAT}}
is written in Python and calls an external SAT solver;
we made the same choice as the authors of the encoding---the sequential version of Glucose 4.0 (which is based on
MiniSAT \cite{DBLP:conf/sat/EenS03}).
Our program and Glucose were compiled with GCC at optimisation level \texttt{-O3}.
Both our algorithm and the program for SAT encoding
are single-threaded.  A time limit of 1000 seconds per instance was used.  We verified that
all solvers that solved an instance within this time limit returned the same treedepth.

Following Ganian et al.\ \cite{DBLP:conf/alenex/GanianLOS19,DBLP:journals/corr/abs-1911-12995}, we used three classes of instances---famous
named graphs (many of which are regular and highly symmetrical), standard graphs (binary trees,
cliques, complete bipartite graphs, cycle graphs, path graphs, and square grids), and random graphs.

\begin{table}[htb]
\centering
 \begin{tabular}{l r r r r r r r r} 
 \toprule
 Instance & $n$ & $m$ & $\mathit{td}$ & All & $-$LB & $-$Sym & $-$Dom & SAT \\ [0.5ex] 
 \midrule
Diamond & 4 & 5 & 3 & 0.002 & 0.002 & 0.002 & 0.002 & 0.002 \\
Bull & 5 & 5 & 3 & 0.003 & 0.003 & 0.002 & 0.002 & 0.041 \\
Butterfly & 5 & 6 & 3 & 0.003 & 0.003 & 0.002 & 0.003 & 0.045 \\
Prism & 6 & 9 & 5 & 0.002 & 0.002 & 0.003 & 0.002 & 0.037 \\
Moser & 7 & 11 & 5 & 0.002 & 0.002 & 0.002 & 0.003 & 0.055 \\
Wagner & 8 & 12 & 6 & 0.002 & 0.002 & 0.002 & 0.003 & 0.067 \\
Pmin & 9 & 12 & 5 & 0.002 & 0.002 & 0.002 & 0.002 & 0.100 \\
Petersen & 10 & 15 & 6 & 0.002 & 0.002 & 0.002 & 0.002 & 0.129 \\
Goldner & 11 & 27 & 5 & 0.002 & 0.003 & 0.002 & 0.002 & 0.195 \\
Grotzsch & 11 & 20 & 7 & 0.003 & 0.003 & 0.003 & 0.002 & 0.187 \\
Herschel & 11 & 18 & 5 & 0.002 & 0.002 & 0.002 & 0.002 & 0.194 \\
Chvatal & 12 & 24 & 8 & 0.003 & 0.004 & 0.009 & 0.003 & 0.402 \\
Durer & 12 & 18 & 7 & 0.002 & 0.003 & 0.003 & 0.002 & 0.262 \\
Franklin & 12 & 18 & 7 & 0.002 & 0.002 & 0.004 & 0.003 & 0.273 \\
Frucht & 12 & 18 & 6 & 0.003 & 0.003 & 0.003 & 0.002 & 0.248 \\
Tietze & 12 & 18 & 7 & 0.003 & 0.003 & 0.003 & 0.002 & 0.266 \\
Paley13 & 13 & 39 & 10 & 0.003 & 0.005 & 0.428 & 0.003 & 2.931 \\
Poussin & 15 & 39 & 9 & 0.005 & 0.009 & 0.101 & 0.005 & 2.478 \\
Clebsch & 16 & 40 & 10 & 0.005 & 0.020 & 0.974 & 0.004 & 14.147 \\
Hoffman & 16 & 32 & 8 & 0.003 & 0.010 & 0.013 & 0.003 & 1.500 \\
Shrikhande & 16 & 48 & 11 & 0.010 & 0.027 & 13.471 & 0.010 & 51.579 \\
Sousselier & 16 & 27 & 8 & 0.004 & 0.017 & 0.017 & 0.004 & 1.263 \\
Errera & 17 & 45 & 10 & 0.007 & 0.022 & 2.486 & 0.006 & 16.225 \\
Paley17 & 17 & 68 & 14 & 0.056 & 0.072 & * & 0.052 & * \\
Pappus & 18 & 27 & 8 & 0.003 & 0.029 & 0.019 & 0.003 & 2.363 \\
Robertson & 19 & 38 & 10 & 0.018 & 0.365 & 3.441 & 0.021 & 43.926 \\
Desargues & 20 & 30 & 9 & 0.004 & 0.097 & 0.323 & 0.005 & 15.208 \\
Dodecahedron & 20 & 30 & 9 & 0.005 & 0.104 & 0.329 & 0.004 & 11.871 \\
FlowerSnark & 20 & 30 & 9 & 0.008 & 0.298 & 0.311 & 0.007 & 13.415 \\
Folkman & 20 & 40 & 9 & 0.004 & 0.056 & 0.118 & 0.007 & 10.071 \\
Brinkmann & 21 & 42 & 11 & 0.195 & 3.637 & * & 0.183 & * \\
Kittell & 23 & 63 & 12 & 0.405 & 3.094 & * & 0.559 & * \\
McGee & 24 & 36 & 11 & 0.219 & 24.042 & 344.762 & 0.175 & * \\
Nauru & 24 & 36 & 10 & 0.056 & 4.914 & 15.565 & 0.048 & 179.968 \\
Holt & 27 & 54 & 13 & 6.680 & 441.213 & * & 5.623 & * \\
WatkinsSnark & 50 & 75 & 13 & * & * & * & 870.345 & * \\
B10Cage & 70 & 105 & & * & * & * & * & * \\
Ellingham & 78 & 117 & & * & * & * & * & * \\
 \bottomrule
 \end{tabular}
 \caption{{Solving times in seconds for famous graphs.  An asterisk indicates timeout at 1000 s.}}
 \label{table:famous}
\end{table}

\subparagraph*{Famous graphs.}
Table \ref{table:famous} shows run times in seconds for famous graphs. 
The second and third columns show the number of vertices and edges in each graph, and the fourth
column shows the treedepth if it is known.  The next
four columns show run times for our algorithm; ``All'' has all features enabled, while
``$-$LB'', ``$-$Sym'', and ``$-$Dom'' have lower bounding, symmetry breaking, and domination rules
turned off respectively.  The final column shows run times for the partition-based SAT encoding
\cite{DBLP:conf/alenex/GanianLOS19,DBLP:journals/corr/abs-1911-12995}.  With all features on, our algorithm typically performs orders
of magnitude faster than the SAT encoding.  Both the symmetry breaking and lower bound
features contribute to the algorithm's performance, but on this set of benchmark instances, the domination
rule slows the algorithm down slightly.  With
the rule switched off, the algorithm closes two open instances---the Holt and Watkins Snark graphs---both
of which have treedepth 13.

\subparagraph*{Standard graphs.}
Table \ref{table:standard} shows run times in seconds for standard instances.
Again, our algorithm is typically much faster than the SAT encoding.  The
usefulness of the domination rule is demonstrated on these instances; with it
switched off, the larger clique and bipartite instances could not be solved
within the time limit.  The lower bounding and symmetry breaking features also
improve the run time on square grid graphs.

\begin{table}[htb]
\centering
 \begin{tabular}{l r r r r r r r r} 
 \toprule
 Instance & $n$ & $m$ & $\mathit{td}$ & All & $-$LB & $-$Sym & $-$Dom & SAT \\ [0.5ex] 
 \midrule
Binary tree 10 & 10 & 9 & 3 & 0.003 & 0.003 & 0.002 & 0.002 & 0.088 \\
Binary tree 20 & 20 & 19 & 4 & 0.002 & 0.003 & 0.002 & 0.002 & 0.804 \\
Binary tree 30 & 30 & 29 & 5 & 0.003 & 0.004 & 0.002 & 0.003 & 4.319 \\
Binary tree 40 & 40 & 39 & 5 & 0.002 & 0.008 & 0.002 & 0.002 & 19.021 \\
Binary tree 50 & 50 & 49 & 5 & 0.002 & 0.040 & 0.002 & 0.003 & 52.950 \\
\addlinespace[0.5em]
Clique10 & 10 & 45 & 10 & 0.003 & 0.002 & 0.003 & 0.003 & 0.003 \\
Clique20 & 20 & 190 & 20 & 0.003 & 0.003 & 0.002 & 0.879 & 0.006 \\
Clique30 & 30 & 435 & 30 & 0.003 & 0.003 & 0.003 & * & 0.220 \\
Clique40 & 40 & 780 & 40 & 0.003 & 0.004 & 0.003 & * & 0.021 \\
Clique50 & 50 & 1225 & 50 & 0.005 & 0.004 & 0.004 & * & 0.035 \\
\addlinespace[0.5em]
Complete bipartite 10 & 10 & 25 & 6 & 0.003 & 0.003 & 0.003 & 0.002 & 0.132 \\
Complete bipartite 20 & 20 & 100 & 11 & 0.002 & 0.003 & 0.003 & 0.012 & 97.689 \\
Complete bipartite 30 & 30 & 225 & 16 & 0.003 & 0.004 & 0.008 & 20.394 & * \\
Complete bipartite 40 & 40 & 400 & 21 & 0.005 & 0.005 & 0.194 & * & * \\
Complete bipartite 50 & 50 & 625 & 26 & 0.007 & 0.010 & 7.565 & * & * \\
\addlinespace[0.5em]
Cycle 10 & 10 & 10 & 5 & 0.002 & 0.003 & 0.003 & 0.003 & 0.137 \\
Cycle 20 & 20 & 20 & 6 & 0.002 & 0.004 & 0.002 & 0.003 & 2.194 \\
Cycle 30 & 30 & 30 & 6 & 0.002 & 0.007 & 0.002 & 0.002 & 16.151 \\
Cycle 40 & 40 & 40 & 7 & 0.002 & 0.409 & 0.002 & 0.002 & 67.870 \\
Cycle 50 & 50 & 50 & 7 & 0.003 & 0.762 & 0.002 & 0.002 & 236.847 \\
\addlinespace[0.5em]
Path 10 & 10 & 9 & 4 & 0.003 & 0.003 & 0.003 & 0.003 & 0.146 \\
Path 20 & 20 & 19 & 5 & 0.002 & 0.003 & 0.003 & 0.002 & 2.146 \\
Path 30 & 30 & 29 & 5 & 0.002 & 0.010 & 0.003 & 0.003 & 15.688 \\
Path 40 & 40 & 39 & 6 & 0.003 & 0.181 & 0.003 & 0.003 & 67.505 \\
Path 50 & 50 & 49 & 6 & 0.002 & 0.405 & 0.002 & 0.002 & 251.987 \\
\addlinespace[0.5em]
Square grid $2\times2$ & 4 & 4 & 3 & 0.002 & 0.003 & 0.003 & 0.003 & 0.033 \\
Square grid $3\times3$ & 9 & 12 & 5 & 0.003 & 0.003 & 0.003 & 0.002 & 0.100 \\
Square grid $4\times4$ & 16 & 24 & 7 & 0.003 & 0.004 & 0.003 & 0.002 & 0.841 \\
Square grid $5\times5$ & 25 & 40 & 9 & 0.025 & 2.219 & 0.797 & 0.026 & 17.869 \\
Square grid $6\times6$ & 36 & 60 & 11 & 1.363 & * & * & 1.619 & * \\
 \bottomrule
 \end{tabular}
 \caption{{Solving times in seconds for standard graphs.  An asterisk indicates timeout at 1000 s.}}
 \label{table:standard}
\end{table}

\subparagraph*{Random graphs.}
We generated random graphs using the Erd\H{o}s-Rényi $G(n,p)$ model, with
$n \in \{12, 16, 20\}$ vertices and edge probabilities
$p \in \{0.1, 0.2, \dots, 0.9\}$.  Ten instances were generated for each $n,p$ pair.
Our algorithm solved each of the 270 instances in less than 0.3 seconds per instance; the SAT
encoding exceeded the time limit of 1000 seconds on 53 of the 90 instances with 20 vertices.

%Tables \ref{table:random12} to \ref{table:random20} show results for random instances.
%The density is shown in the header row, and the values in the body of the table show
%how many instances out of 10 could be solved within the time limit.
%
%\input{../new-experiments/results/random12.tex}
%
%\input{../new-experiments/results/random16.tex}
%
%\input{../new-experiments/results/random20.tex}

\subparagraph*{Summary of experimental results.}
\Cref{fig:runtimesscatter} summarises, for all instances, the run times of our algorithm
and the SAT encoding.  Each point shows the two algorithms' run times
for a single instance.  Timeouts are shown as 1000 seconds.  For the harder instances that take more
than ten seconds to solve with the SAT encoding, our algorithm is typically around three orders
of magnitude faster.

\begin{figure}[htb]
  \centering
  \includegraphics[scale=.95]{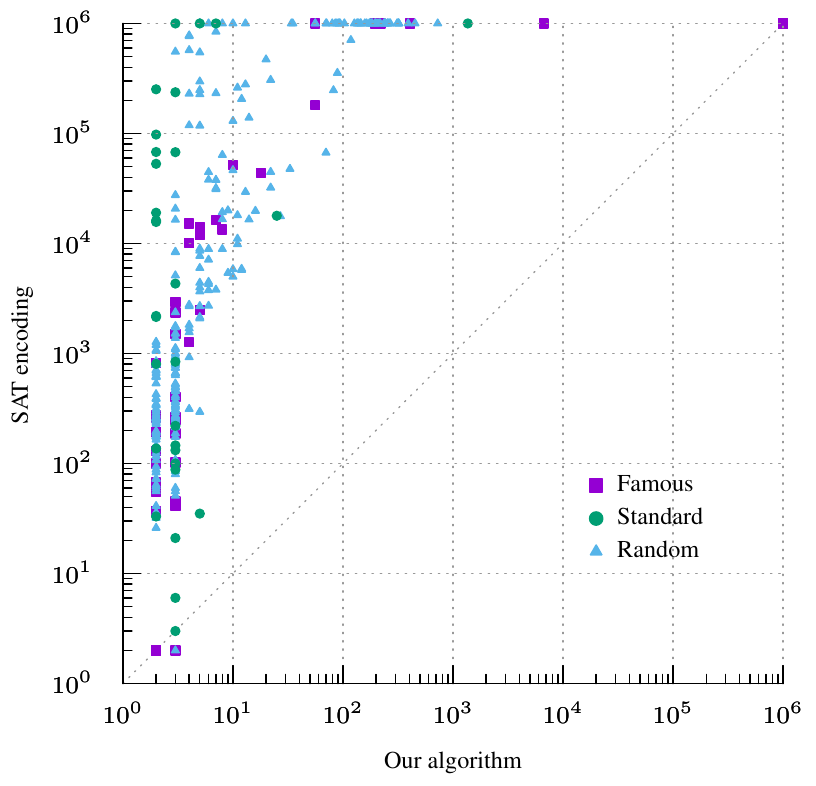}
  \caption{Run times in ms of our algorithm and the SAT encoding. Each point represents one instance.}
  \label{fig:runtimesscatter}
\end{figure}

We end this section by noting that the SAT-encoding program writes each SAT
instance it generates to disk, whereas our program performs no disk I/O while
solving.  This implementation detail contributes to the run times of the
SAT-based program.  For the famous graphs, the cost of this disk I/O
does not meaningfully affect our comparison of run times; on
the three most difficult famous instances that can be solved by the SAT-based
program (Nauru, Shrikhande and Robertson), over $95\%$ of the total run time is
spent within the SAT solver on a single unsatisfiable instance of the SAT
problem.  For some of the standard and random graphs, such as complete
bipartite graphs, a similar pattern holds.  But for other graphs in these
classes, such as binary trees, the SAT-based program spends most of its time
creating encodings and writing them to disk, and for these instances it
is likely that the program could be improved significantly by
avoiding writing to disk.

\section{Conclusion} \label{sec:conclusion}

We have introduced an algorithm for computing the exact treedepth of a graph, and shown
experimentally that it runs orders of magnitude faster than the current state of the art
on a varied set of benchmark instances.  The core of the algorithm is a
simple pair of mutually-recursive functions.
To this basic algorithm, we have added symmetry breaking, domination,
lower bounding, and vertex ordering rules.  There is room for further
improvement to each of these extensions of the algorithm.

There is also scope for improvement in finding a good upper bound on the treedepth
quickly.
SAT encodings of the problem have advantages in this regard:
modern SAT solvers implement restarts and good heuristics, both of which help to find
good solutions quickly.  Future research could combine the benefits of SAT solvers with
the techniques introduced in this paper, either by including some of the techniques
in a SAT model or by adding features such as periodic restarts to our algorithm.

\FloatBarrier

\bibliography{bib}

\appendix
\section{Proof of Correctness} \label{appendix:proof}

\begin{proposition}\label{correctnessproposition}
  The function $\TDOptimise()$ returns the treedepth of the input graph $G$.
\end{proposition}
\begin{proof}
If $G$ is the empty graph, the last line of $\TDGeneral()$ correctly returns $\AlgVar{true}$, since
  the graph has no connected components.

  For non-empty graphs, we prove by induction on $|V(G)|$ that $\TDGeneral()$ and $\TDConnected()$
  give correct results.  In the base case, consider calls to each function with a single-vertex graph.
  The function $\TDConnected()$ correctly returns $\AlgVar{true}$ at \codelineref{EndSingleVertexGraph}.
  The function $\TDGeneral()$ returns $\AlgVar{false}$ at \codelineref{TDGeneralK0} if $k=0$; otherwise
  $\TDConnected()$ is called for the single-vertex component and
  $\AlgVar{true}$ is returned.

  In the inductive case, assume that $\TDGeneral()$ and $\TDConnected()$ give correct results for
  graphs with fewer than $|V(G)|$ vertices.  
  The function $\TDConnected()$ is correct, since the loop beginning at \codelineref{VLoop}
  returns $\AlgVar{true}$ if and only if 
  there is there is some vertex $v$ whose removal leaves a graph which has an elimination forest of
  depth at most $k-1$.
  To show the correctness of $\TDGeneral()$, we consider three cases.  If $k=0$, \codelineref{TDGeneralK0} returns
  the correct answer.  Otherwise, if $G$ is disconnected, $\AlgVar{true}$ is returned if and only
  if an elimination tree of depth no more than $k$ can be found for each connected component on
  \codelineref{LoopToCallTDConnected}.
  Since each component has fewer than $|V(G)|$ vertices, the results from $\TDConnected()$ are correct
  by our inductive assumption.
  In the final case, if $G$ is connected, the function returns the result of
  $\TDConnected(G,k,\dots)$, which we have already shown to be correct.

  Given the correctness of the first two functions, the function $\TDOptimise()$ finds,
  as required, the lowest value $k$ such that $G$ has an elimination forest of depth $k$.
\end{proof}

\end{document}